\documentclass[%
 reprint,
 amsmath,amssymb,
 aps,
]{revtex4-1}

\usepackage{color}
\usepackage[euler]{textgreek}
\usepackage{kotex}
\usepackage{graphicx}% Include figure files
\usepackage{dcolumn}% Align table columns on decimal point
\usepackage{bm}% bold math
\newtheorem{theorem}{Theorem}
\newtheorem{definition}{Definition}
\newtheorem{example}{Example}

\newtheorem{proposition}{Proposition}
\newenvironment{proof}[1][Proof]{\textbf{#1.} }{\  \rule{0.5em}{0.5em}}
\makeatletter
\def \@removefromreset#1#2{\let \@tempb \@elt
\def \@tempa#1{@&#1}\expandafter \let \csname @*#1*\endcsname \@tempa
\def \@elt##1{\expandafter \ifx \csname @*##1*\endcsname \@tempa \else
\noexpand \@elt{##1}\fi}     \expandafter \edef \csname cl@#2\endcsname{\csname cl@#2\endcsname}     \let \@elt \@tempb
\expandafter \let \csname @*#1*\endcsname \@undefined}

\@removefromreset{equation}{section}

\@removefromreset{theorem}{section}
\makeatother
\begin{document}

\preprint{APS/123-QED}

\title{Enhanced Optimal Quantum Communication \\ by Generalized Phase Shift Keying Coherent Signal }% Force line breaks with \\

\author{Min Namkung}
\author{Jeong San Kim}%
 \email{freddie1@khu.ac.kr}
\affiliation{%
 Department of Applied Mathematics and Institute of Natural Sciences, Kyung Hee University, Yongin 17104, Republic of Korea
}%

\date{\today}% It is always \today, today,
             %  but any date may be explicitly specified

\begin{abstract}
It is well known that the maximal success probability of the binary quantum communication can be improved by using a sub-Poissonian non-standard coherent state as an information carrier. In the present article, we consider the quantum communication with $N$-ary phase shift keying ($N$-PSK) signal for an arbitrary positive integer $N>1$. By using non-standard coherent state, we analytically provide the maximal success probability of the quantum communication with $N$-PSK. Unlike the binary case, we show that even super-Poissonianity of non-standard coherent state can improve the maximal success probability of $N$-PSK quantum communication.
\end{abstract}

%\keywords{Suggested keywords}%Use showkeys class option if keyword
                              %display desired
\maketitle

%\tableofcontents

\section{Introduction}
In optical communication, a sender encodes a message in an optical signal and sends it to a receiver who detectes the signal to decode the  message \cite{g.cariolaro}. Thus, the success probability of the optical communication is determined by the physical and statistical properties of the optical signal together with the structure of the receiver's measurement device. In classical optical communication, the receiver can use an on-off detector to decode a sender's message encoded in on-off keying signal \cite{c.w.helstrom,k.tsujino}, and a homodyne detector for binary phase shift keying signal \cite{j.g.proakis}. However, the maximal success probability for decoding encoded messages by using \textit{conventional measurements} such as the on-off and the homodyne detectors cannot exceed the standard quantum limit.

One of the goals in quantum communication is to design a novel measurement so that the maximal success probability to decode messages can surpass standard quantum limit \cite{i.a.burenkov}. According to the quantum theory, optical signal is described as a density operator on a Hilbert space and a measurement is descrived as a positive-operator-valued measure (POVM), therefore the quantum communication is described as a quantum state discrimination protocol \cite{s.m.barnett,j.a.bergou}. 

Minimum error discrimination \cite{j.bae,d.ha} is one representative state discrimination strategy used in various quantum communication protocols. When one bit message is encoded by binary coherent states, minimum error discrimination between the binary coherent states can be implemented via the Dolinar receiver \cite{s.j.dolinar}. However, when several bits are encoded and sequentially sent, the photon number detector used for the discrimination may not efficiently react along the received states \cite{i.a.burenkov}. For this reason, $N$-ary coherent states such as $N$-amplitude shift keying ($N$-ASK) signal \cite{c.w.helstrom} and $N$-phase shift keying ($N$-PSK) signal \cite{j.g.proakis} have been considered to send $\log_2N$ bit messages.

According to a recent work \cite{e.m.f.curado}, the maximal success probability (or Helstrom bound) of discriminating a message encoded in 2-PSK signal composed of \textit{non-standard coherent states (NS-CS)} with a novel quantum measurement can be improved by the sub-Poissonianity of the NS-CS. Moreover, the experimental method for implementing the quantum measurement reaching for the Helstrom bound has recently been proposed \cite{m.namkung}. Since the negative Mandel parameter to quantify the sub-Poissonianity is considered as a resource in a non-classical light \cite{s.dey}, this result implies that the sub-Poissonianity can be a resource for improving the performance of the quantum communication. 

In the present article, we consider the quantum communication with $N$-PSK signal for arbitrary an arbitrary positive integer $N>1$. By using non-standard coherent state, we analytically provide the maximal success probability of the quantum communication with $N$-PSK. Unlike the binary case, we show that even super-Poissonianity of non-standard coherent state can improve the maximal success probability of $N$-PSK quantum communication: The Helstrom bound can be improved by considering the sub-Poissonian NS-CS for $N=3$, meanwhile the super-Poissonian NS-CS can improve the Helstrom bound for $N=4$ and $N=8$.

For $N>2$, $N$-PSK signal enables us to transmit a $\log_{2}N$-bit message per a signal pulse, which is a better information exchange rate than binary-PSK. Moreover it is also known that $N$-PSK signal can provide an improved information exchange rate between the sender and receiver even though the receiver's measurement is slow \cite{i.a.burenkov}. However, the maximal success probability of discriminating a message encoded in $N$-PSK signal generally decreases as $N$ is getting large. Thus our results about the possible enhancement of the maximal success probability in $N$-PSK quantum communication by NS-CS is important and even necessary to design efficient quantum communication schemes.

The present article is organized as follows. In Section 2, we briefly review the problem of minimum error discrimination among $N$ symmetric pure states. In Section 3, we provide the analytical Helstrom bound of $N$-PSK signal composed of NS-CS. In Section 4, we investigate the Helstrom bound of $N$-PSK signal composed of optical spin coherent states (OS-CS), Perelomov coherent states (P-CS), Barut-Girardello coherent states (BG-CS) and modified Susskind-Glogower coherent states (mSG-CS), and discuss the relation between the sub-Poissonianity of the non-classical light and the performance of the $N$-PSK quantum communication. Finally, in Section 5, we propose the conclusion of the present article.

\section{Preliminaries: Minimum Error Discrimination among Symmetric Pure States} 

In quantum communication, Alice (sender) prepares her message $x\in\{1,\cdots,N\}$ with a prior probability $q_x\in\{q_1,\cdots,q_N\}$, encodes the message in a quantum state $\rho_x\in\{\rho_1,\cdots,\rho_N\}$, and sends the quantum state to Bob (receiver). Bob performs a quantum measurement described as a POVM $\{M_1,\cdots,M_N\}$ to discriminate the encoded message. In the POVM, $M_x$ is a POVM element with respect to a result $x$.  

For a given ensemble $\mathcal{E}=\{q_x,\rho_x\}_{x=1}^{N}$ of Alice and a POVM $\mathcal{M}=\{M_x\}_{x=1}^{N}$ of Bob, the success probability of the quantum communication between Alice and Bob is described by the success probability of the state discrimination,
\begin{equation}
P_s(\mathcal{E},\mathcal{M})=\sum_{x=1}^{N}q_x\mathrm{tr}\left\{\rho_xM_x\right\},\label{success_probability}
\end{equation}
One way to optimize the efficiency of quantum communication is to consider a POVM that maximizes the success probability in Eq. (\ref{success_probability}). In this case, the maximization of the success probability in Eq. (\ref{success_probability}) is equivalent to the minimization of the error probability
\begin{equation}
P_e(\mathcal{E},\mathcal{M})=1-P_s(\mathcal{E},\mathcal{M})=\sum_{x=1}^{N}\sum_{y\not=x}q_x\mathrm{tr}\left\{\rho_xM_y\right\}.\label{error_probability}
\end{equation}
\textit{Minimum error discrimination} is to minimize the error probability in Eq. (\ref{error_probability}) over all possible POVMs $\mathcal{M}$ of Bob.

For a given ensemble $\mathcal{E}$, it is known that the following inequality is a necessary and sufficient condition for POVM $\mathcal{M}$ minimizing the error probability \cite{c.w.helstrom,s.m.barnett2},
\begin{eqnarray}
\sum_{z=1}^{N}q_z\rho_zM_z-q_x\rho_x\ge 0, \ \ \forall x\in\{1,\cdots,N\}.\label{inequality_condition}
\end{eqnarray}
Moreover, it is known that the following equality is a useful necessary condition to characterize the structure of the POVM,
\begin{eqnarray}
M_x(q_x\rho_x-q_y\rho_y)M_y=0, \ \ \forall x,y\in\{1,\cdots,N\}.\label{equality_condition}
\end{eqnarray}
If every quantum state $\rho_x$ is pure, that is, $\rho_x=|\psi_x\rangle\langle\psi_x|$, the optimal POVM is given by a rank-1 projective measurement \cite{c.w.helstrom}. In other word, $M_x=|\pi_x\rangle\langle\pi_x|$ for every $x\in\{1,\cdots,N\}$. 

Now, we focus on the minimum error discrimination among a specific class of pure states, called \textit{symmetric pure states}. 
\begin{definition}
\cite{a.chefles} For a positive integer $N$, the distinct pure states $|\psi_1\rangle,\cdots,|\psi_N\rangle$ are called \textit{symmetric},
if there exists a unitary operator $V$ such that 
\begin{equation}
|\psi_x\rangle=V^{x-1}|\psi_1\rangle 
\label{symuni}
\end{equation}
for $x= 1,2,\cdots,N$ and
\begin{equation}
V^N=\mathbb{I},
\end{equation}
where $\mathbb{I}$ is an identity operator on a subspace spanned by $\{|\psi_1\rangle,\cdots,|\psi_N\rangle\}$.
\end{definition}

The Gram matrix composed of the symmetric pure states in Definition 1 is 
\begin{equation}
G=\left(\langle\psi_x|\psi_y\rangle\right)_{x,y=1}^{N}.\label{gram}
\end{equation}
From a straightforward calculation, the eigenvalues of the Gram matrix in Eq. (\ref{gram}) are in forms of
\begin{equation}
\lambda_p=\sum_{k=1}^{N}\langle\psi_j|\psi_k\rangle e^{-\frac{2\pi i (p-1)(j-k)}{N}}, \ \ p=1,2,\cdots,N,\label{lambda}
\end{equation}
for any choice of $j\in\{1,2,\cdots,N\}$. We note that the set $\{\lambda_p\}_{p=1}^{N}$ is invariant under the choice of $j$ due to the symmetry of the pure state $\{|\psi_1\rangle,\cdots,|\psi_N\rangle\}$. The following proposition provides the maximal success probability of the minimum error discrimination among the symmetric pure states in Definition 1. 
\begin{proposition}
\cite{c.w.helstrom} Let $\mathcal{E}_{sym}$ be an equiprobable ensemble of symmetric pure states $|\psi_1\rangle,\cdots,|\psi_N\rangle$. Then, the maximal success probability is given as
\begin{equation}
P_{hel}(\mathcal{E}_{sym})=\frac{1}{N^2}\left(\sum_{p=1}^{N}\sqrt{\lambda_p}\right)^2,\label{Helstrom_bound}
\end{equation}
where $\lambda_p$ are the eigenvalues of the Gram matrix composed of $\{|\psi_1\rangle,\cdots,|\psi_N\rangle\}$ in Eq. (\ref{gram}).
\end{proposition}

Eq. (\ref{Helstrom_bound}) is also called \textit{Helstrom bound}, and $1-P_{hel}(\mathcal{E}_{sym})$ is called \textit{minimum error probability}.

\section{Optimal Communication with \\ Phase Shift Keying (PSK) Signal}
In quantum optical communication, phase shift keying (PSK) signal is expressed as equiprobable symmetric pure states \cite{g.cariolaro}. In this section, we derive the maximal success probability of the quantum communication with PSK signal composed of generalized coherent states. First, we provide definition of the generalized coherent state which encapsulates standard coherent state (S-CS) and non-standard coherent state (NS-CS) as special cases.
\begin{definition}
\cite{e.m.f.curado} If a pure state takes the form
\begin{equation}
|\alpha,\vec{h}\rangle=\sum_{n=0}^{\infty}\alpha^nh_n(|\alpha|^2)|n\rangle, \ \ \alpha\in\mathbb{C},\label{generalized_coherent_state}
\end{equation}
where $\{|n\rangle|n\in\mathbb{Z}^+\cup\{0\}\}$ is Fock basis and $\vec{h}$ is a tuple of real-valued functions $h_n:[0,R^2]\rightarrow\mathbb{R}$ satisfying
\begin{eqnarray}
&&\sum_{n=0}^{\infty}u^n\left\{h_n(u)\right\}^2=1,\\
&&\sum_{n=0}^{\infty}nu^n\left\{h_n(u)\right\}^2 \\
&& \ \ \ \ \ \ \ \ \textrm{ is a strictly increasing function of }u,\nonumber\\
&&\int_0^{R^2}duw(u)u^n\left\{h_n(u)\right\}^2=1 \\
&& \ \ \ \ \ \ \ \  \textrm{ for a real-valued function }w:[0,R^2]\rightarrow\mathbb{R}^+,\nonumber
\end{eqnarray} 
then the pure state is called \textit{generalized coherent state}. If every real-valued function $h_n$ in Eq. (\ref{generalized_coherent_state}) takes the form
\begin{equation}
h_n(u)=\frac{1}{\sqrt{n!}}e^{-\frac{1}{2}u}, \ \ \forall n\in\mathbb{Z}^+\cup\{0\},\label{special_case}
\end{equation}
then Eq. (\ref{generalized_coherent_state}) is called \textit{standard coherent state} (S-CS) \cite{r.j.glauber}. Otherwise, Eq. (\ref{generalized_coherent_state}) is called \textit{non-standard coherent state} (NS-CS).
\end{definition}

Several examples of NS-CS have been introduced such as optical spin coherent state (OS-CS) \cite{a.m.perelomov}, Perelomov coherent state (P-CS) \cite{a.m.perelomov}, Barut-Girardello coherent state (BG-CS) \cite{a.o.barut} and modified Susskind-Glogower coherent state (mSG-CS) \cite{j.-p.gazeau}.
\begin{example}
For a given non-negative integer $\widetilde{n}$, if $h_n$ takes the form 
\begin{eqnarray}
h_n(u)=\sqrt{\frac{\widetilde{n}!}{n!(\widetilde{n}-n)!}}(1+u)^{-\frac{\widetilde{n}}{2}},
\end{eqnarray}
for $0\le n \le \widetilde{n}$ and $h_n(u)=0$ for $n>\widetilde{n}$, then the generalized coherent state in Eq. (\ref{generalized_coherent_state}) is called \textit{optical spin coherent state} (OS-CS).
\end{example}
\begin{example} 
For all non-negative integer $n$ and a real number $\varsigma$ with $\varsigma\ge1/2$, if $h_n$ takes the form 
\begin{eqnarray}
h_n(u)=\frac{1}{\mathcal{N}(u)}\sqrt{\frac{\Gamma(2\varsigma)}{n!\Gamma(2\varsigma+n)}},
\end{eqnarray}
then the generalized coherent state in Eq.(\ref{generalized_coherent_state}) is called \textit{Barut-Girardello coherent state} (BG-CS). Here, $\Gamma$ is the Gamma function of the first kind and $\mathcal{N}(u)$ is a normalization factor
\begin{equation}
\mathcal{N}(u)=\Gamma(2\varsigma)u^{1/2-u}I_{2\varsigma-1}(2\sqrt{u}),
\end{equation}
where $I_\nu$ is the modified Bessel function of the first kind.
\end{example}
\begin{example}
For all non-negative integer $n$, if $h_n$ takes the form 
\begin{eqnarray}
h_n(u)=\sqrt{\frac{n+1}{\mathcal{\bar{N}}(u)}}\frac{1}{u^{\frac{n+1}{2}}}J_{n+1}(2\sqrt{u}),
\end{eqnarray}
then the generalized coherent state in Eq. (\ref{generalized_coherent_state}) is called \textit{modified Susskind-Glogower coherent state} (mSG-CS). Here, $J_n$ is the Bessel function of the first kind and $\mathcal{\bar{N}}(u)$ is a normalization factor
\begin{eqnarray}
\mathcal{\bar{N}}(u)&=&\frac{1}{u}\Big[2u\left\{J_0(2\sqrt{u})\right\}^2\\
&-&\sqrt{u}J_0(2\sqrt{u})J_1(2\sqrt{u})+2u\left\{J_1(2\sqrt{u})\right\}^2\Big].\nonumber
\end{eqnarray}
\end{example}
\begin{example}
For all non-negative integer $n$, $\varsigma$ and an integer or half-integer with $\varsigma\ge1/2$, if $h_n$ takes the form 
\begin{eqnarray}
h_n(u)=\sqrt{\frac{(2\varsigma-1+n)!}{n!(2\varsigma-1)!}}(1-u)^{\varsigma},
\end{eqnarray}
then the generalized coherent state in Eq. (\ref{generalized_coherent_state}) is called \textit{Perelomov coherent state} (P-CS). 
\end{example}

We mainly focus on which NS-CS provided in the examples can give the advantage to the $N$-ary PSK quantum communication. For this reason, we define the $N$-ary generalized PSK ($N$-GPSK) signal as follows.
\begin{definition}
If an equiprobable ensemble $\mathcal{E}_{gcs}$ consists of generalized coherent states, 
\begin{equation}
\left\{|\alpha_x,\vec{h}\rangle|x\in\{1,2,\cdots,N\}\right\},\label{PSK}
\end{equation}
where $N\in\mathbb{Z}^+$ and $\alpha_x\in\mathbb{C}$ such that
\begin{equation}
\alpha_x=\alpha e^{\frac{2\pi i}{N}x},\label{alpha_x}
\end{equation}
with a non-negative integer $\alpha$, then the ensemble $\mathcal{E}_{gcs}$ is called $N$-ary generalized PSK ($N$-GPSK) signal. 
\end{definition}

Moreover, $N$-GPSK signal is called $N$-ary standard PSK ($N$-SPSK) signal \cite{g.cariolaro} if every coherent state in Eq. (\ref{PSK}) is S-CS, and $N$-PSK signal is called $N$-ary non-standard PSK ($N$-NSPSK) signal if every coherent state in Eq. (\ref{PSK}) is NS-CS. The following theorem shows that the generalized coherent states in Definition 3 are symmetric.
\begin{theorem}
For given distinct generalized coherent states $|\alpha_1,\vec{h}\rangle,\cdots,|\alpha_N,\vec{h}\rangle$, there exists a unitary operator $U$ such that
\begin{eqnarray}
|\alpha_x,\vec{h}\rangle=U^{x-1}|\alpha_1,\vec{h}\rangle, \ \ \forall x\in\{1,2,\cdots,N\},\label{thm1}
\end{eqnarray}
for $x=1,2,\cdots,N$ and
\begin{equation}
U^N=\mathbb{I},\label{U_N}
\end{equation}
where $\mathbb{I}$ is an identity operator on a subspace spanned by $\{|\alpha_1,\vec{h}\rangle,\cdots,|\alpha_N,\vec{h}\rangle\}$.
\end{theorem}
\begin{proof}
Consider a unitary operator
\begin{eqnarray}
U=e^{\frac{2\pi i}{N}a^\dagger a},\label{U}
\end{eqnarray}
where $a$ and $a^\dagger$ are the annihilation and creation operators satisfying
\begin{eqnarray}
&&a|n\rangle=\sqrt{n}|n-1\rangle, \ \ \forall n\in\mathbb{Z}^+,\\
&&a^\dagger|n\rangle=\sqrt{n+1}|n+1\rangle, \ \ \forall n\in\mathbb{Z}^+\cup\{0\},
\end{eqnarray}
respectively. It is straightforward to show that the unitary operator $U$ in Eq. (\ref{U}) satisfies Eq. (\ref{U_N}). 

We also note that
\begin{eqnarray}
U|n\rangle=e^{\frac{2\pi i}{N}n}|n\rangle, 
\end{eqnarray}
for any non-negative integer $n$, therefore we have that 
\begin{eqnarray}
U|\alpha_x,\vec{h}\rangle&=&\sum_{n=0}^{\infty}\alpha_x^nh_n(|\alpha_x|^2)e^{\frac{2\pi i}{N}a^\dagger a}|n\rangle\nonumber\\
&=&\sum_{n=0}^{\infty}\alpha_x^nh_n(|\alpha_x|^2)e^{\frac{2\pi i}{N}n}|n\rangle\nonumber\\
&=&\sum_{n=0}^{\infty}(\alpha_xe^{\frac{2\pi i}{N}})^nh_n(|\alpha_x|^2)|n\rangle,\label{proof}
\end{eqnarray}
for every $x\in\{1,2,\cdots,N-1\}$. Moreover, Eq. (\ref{alpha_x}) leads us to
\begin{eqnarray}
\alpha_xe^{\frac{2\pi i}{N}}=\alpha_{x+1},\label{prop_alpha1}
\end{eqnarray}
for $x\in\{1,2,\cdots,N-1\}$ and
\begin{eqnarray}
|\alpha_x|=\alpha,\label{prop_alpha2}
\end{eqnarray}
for $x\in\{1,2,\cdots,N\}$. From Eqs. (\ref{proof}), (\ref{prop_alpha1}) and (\ref{prop_alpha2}), we have
\begin{eqnarray}
U|\alpha_x,\vec{h}\rangle=\sum_{n=0}^{\infty}(\alpha_{x+1})^nh_n(|\alpha_{x+1}|^2)|n\rangle=|\alpha_{x+1},\vec{h}\rangle.\nonumber\\
\label{finally_proved}
\end{eqnarray}
Eq. (\ref{thm1}) can be shown by an inductive use of Eq. (\ref{finally_proved}), which completes the proof.
\end{proof}

Theorem 1 means that the Helstrom bound of quantum communication with $N$-GPSK signal is given by Eq. (\ref{Helstrom_bound}) in Proposition 1, which is encapsulated in the following theorem. 

\begin{theorem}
The Helstrom bound of $N$-GPSK signal is given by
\begin{eqnarray}
P_{hel}(\mathcal{E}_{gcs})=\frac{1}{N^2}\left(\sum_{p=1}^{N}\sqrt{\lambda_p^{(G)}}\right)^2,\label{helstrom_bound_N}
\end{eqnarray}
where $\lambda_p^{(G)}$ takes the form of
\begin{eqnarray}
\lambda_p^{(G)}=\sum_{k=0}^{M-1}\left[\sum_{n=0}^{\infty}\alpha^{2n}\cos\left\{\frac{2\pi}{N}k(n+p-1)\right\}\left\{h_n(\alpha^2)\right\}^2\right].\nonumber\\
\label{lambda_N}
\end{eqnarray}
for every $p\in\{1,2,\cdots,N\}$.
\end{theorem}
\begin{proof}
For every $j,k\in\{1,2,\cdots,N\}$, the inner product $\langle\alpha_j,\vec{h}|\alpha_k,\vec{h}\rangle$ is
\begin{eqnarray}
\langle\alpha_j,\vec{h}|\alpha_k,\vec{h}\rangle=\sum_{n=0}^{\infty}\left\{\alpha^2e^{i\frac{2\pi}{N}(k-j)}\right\}^n\left\{h_n(\alpha^2)\right\}^2.\label{inner_product_NS_CS}
\end{eqnarray}
From Eq. (\ref{inner_product_NS_CS}) together with Eq. (\ref{lambda}), $\lambda_{p}^{(G)}$ is also obtained by
\begin{eqnarray}
&&\lambda_p^{(G)}\nonumber\\
&&=\sum_{k=1}^{N}\left[\sum_{n=0}^{\infty}\left\{\alpha^2e^{i\frac{2\pi}{N}(k-j)}\right\}^n\left\{h_n(\alpha^2)\right\}^2\right]e^{-\frac{2\pi i (p-1)(j-k)}{N}}\nonumber\\
&&=\sum_{k=1}^{N}\left[\sum_{n=0}^{\infty}\alpha^{2n}e^{i\frac{2\pi}{N}(k-j)(n+p-1)}\left\{h_n(\alpha^2)\right\}^2\right].
\end{eqnarray}
As mentioned before, the set $\{\lambda_p^{(G)}\}_{p=1}^{N}$ is invariant under the choice of $j\in\{1,2,\cdots,N\}$. By choosing $j=1$ and substituting $k$ to $k-1$, $\lambda_p^{(G)}$ can be rewritten by
\begin{eqnarray}
\lambda_p^{(G)}=\sum_{k=0}^{N-1}\left[\sum_{n=0}^{\infty}\alpha^{2n}e^{i\frac{2\pi}{N}k(n+p-1)}\left\{h_n(\alpha^2)\right\}^2\right].\label{lambda_exp}
\end{eqnarray}
Since the Gram matrix is Hermitian, $\lambda_p^{(G)}$ is a real number. Thus, by using the relation
\begin{equation}
\lambda_p^{(G)}=\frac{\lambda_p^{(G)}+\lambda_p^{(G)*}}{2},
\end{equation}
together with Eq. (\ref{lambda_exp}), we have Eq. (\ref{lambda_N}). Due to Theorem 1, every generalized coherent state in $N$-GPSK signal is symmetric. Thus, Proposition 1 and Eq. (\ref{lambda_N}) lead us to the Helstrom bound in Eq. (\ref{helstrom_bound_N}).
\end{proof} 

\section{Sub-Poissonianity of NS-CS and the Helstrom bound}
For $N=3,4$ and 8, we provide illustrative results of the Helstrom bound of $N$-NSPSK signal of Eq. (\ref{helstrom_bound_N}) in case of OS-CS, P-CS, BG-CS and mSG-CS. We also compare these results with the case of $N$-SPSK signal.

\subsection{Optical Spin Coherent States (OS-CS)}
\begin{figure*}
\centering
\includegraphics[scale=0.29]{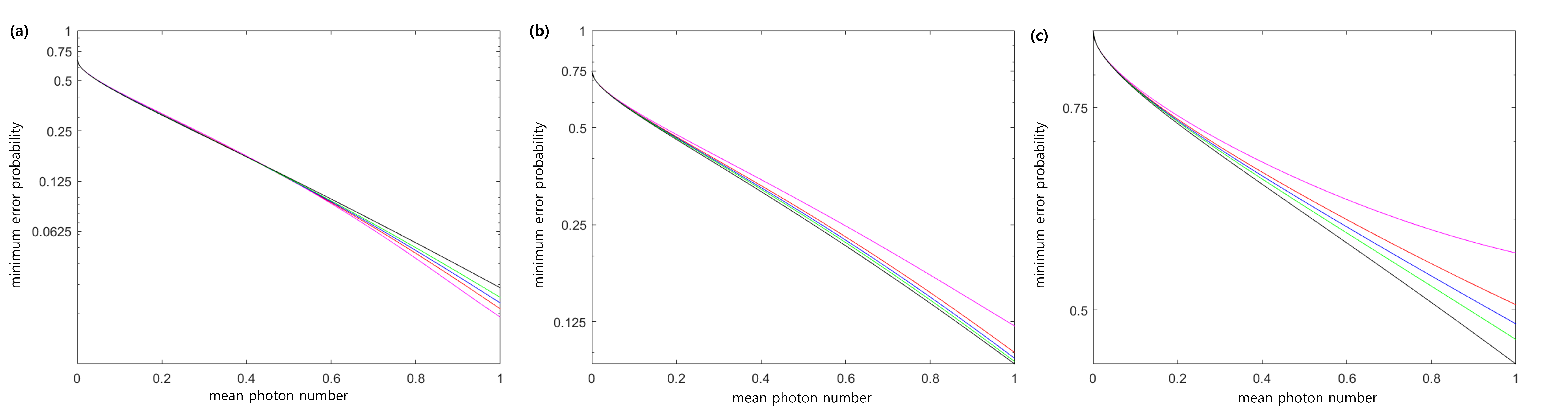}
\caption{The minimum error probabilities of $N$-SPSK signal and $N$-NSPSK signal composed of OS-CS , where (a), (b) and (c) show the case of $N=$3, 4 and 8, respectively. In these figures, purple, red, blue and green lines show the case of $N$-NSPSK signal with $\widetilde{n}=3$, $\widetilde{n}=5$, $\widetilde{n}=7$ and $\widetilde{n}=11$, respectively. Black lines in the figures show the case of $N$-SPSK signal. }
\end{figure*}
\begin{figure*}
\centering
\includegraphics[scale=0.29]{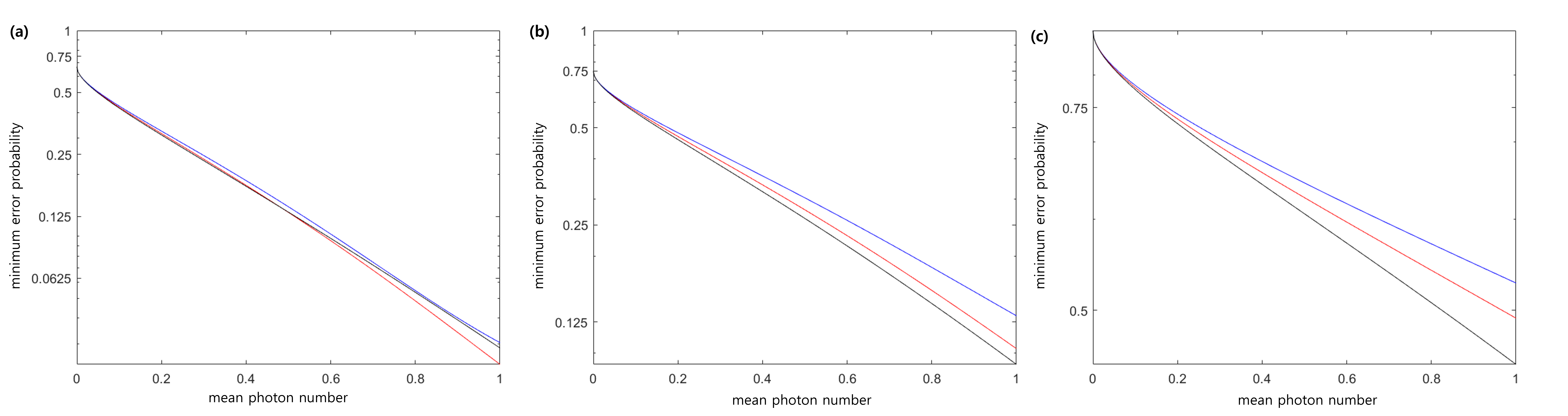}
\caption{The minimum error probabilities of $N$-SPSK signal and $N$-NSPSK signal composed of BG-CS, where (a), (b) and (c) shows the case of $N=$3, 4 and 8, respectively. In these figures, blue and red lines show the case of $N$-NSPSK signal with $\varsigma=0.5$ and $\varsigma=1.5$, respectively. Black lines show the case of $N$-SPSK. }
\end{figure*}
The minimum error probabilities of $N$-SPSK signal and $N$-NSPSK signal composed of OS-CS are illustrated in Fig. 1, where Fig. 1(a), (b) and (c) show the case of $N=$3, 4 and 8, respectively. In these figures, purple, red, blue and green lines show the case of $N$-NSPSK signal with $\widetilde{n}=3$, $\widetilde{n}=5$, $\widetilde{n}=7$ and $\widetilde{n}=11$, respectively. Black lines in the figures show the case of $N$-SPSK signal. 

In Fig. 1(a), the minimum error probabilities of 3-NSPSK signal composed of OS-CS is smaller than that of 3-SPSK signal when mean photon number is large ($\langle n\rangle>0.45$, $\langle n\rangle>0.42$, $\langle n\rangle>0.38$ and $\langle n\rangle>0.37$ in case of $\widetilde{n}=3$, $\widetilde{n}=5$, $\widetilde{n}=7$ and $\widetilde{n}=11$, respectively). In other words, \textit{3-PSK quantum communication can be enhanced by a non-standard coherent state using OS-CS.} However,  in Fig. 1(b), each minimum error probability of 4-NSPSK signal is larger than that of 4-SPSK signal for arbitrary mean photon number. This aspect repeatedly happens in Fig. 1(c) where 8-NSPSK signal is considered. These results imply that \textit{4-PSK and 8-PSK quantum communication cannot be enhanced by OS-CS}.

\subsection{Barut-Girardello Coherent States (BG-CS)}
The minimum error probabilities of $N$-SPSK signal and $N$-NSPSK signal composed of BG-CS are illustrated in Fig. 2, where Fig. 2(a), (b) and (c) shows the case of $N=$3, 4 and 8, respectively. In these figures, blue and red lines show the case of $N$-NSPSK signal with $\varsigma=0.5$ and $\varsigma=1.5$, respectively. Black lines show the case of $N$-SPSK. 

In Fig. 2(a), each minimum error probabilty of 3-NSPSK signal with $\varsigma=1.5$ is smaller than that of 3-SPSK signal when mean photon number is larger than 0.48. Meanwhile, each minimum error probabilty of 3-NSPSK signal with $\varsigma=0.5$ is larger than that of 3-SPSK signal for arbitrary mean photon number. \textit{Thus, enhancing 3-PSK quantum communication by non-standard coherent state using BG-CS depends on the parameter $\varsigma$.} However, in Fig. 2(b), each minimum error probability of $4$-NSPSK signal is larger than that of $4$-SPSK signal for arbitrary mean photon number. This aspect repeatedly happens in Fig. 2(c) where 8-NSPSK signal is considered. These results imply that \textit{4-PSK and 8-PSK quantum communication cannot be enhanced by BG-CS}.

\subsection{Modified Susskind-Glogower Coherent States (mSG-CS)}
\begin{figure*}
\centering
\includegraphics[scale=0.29]{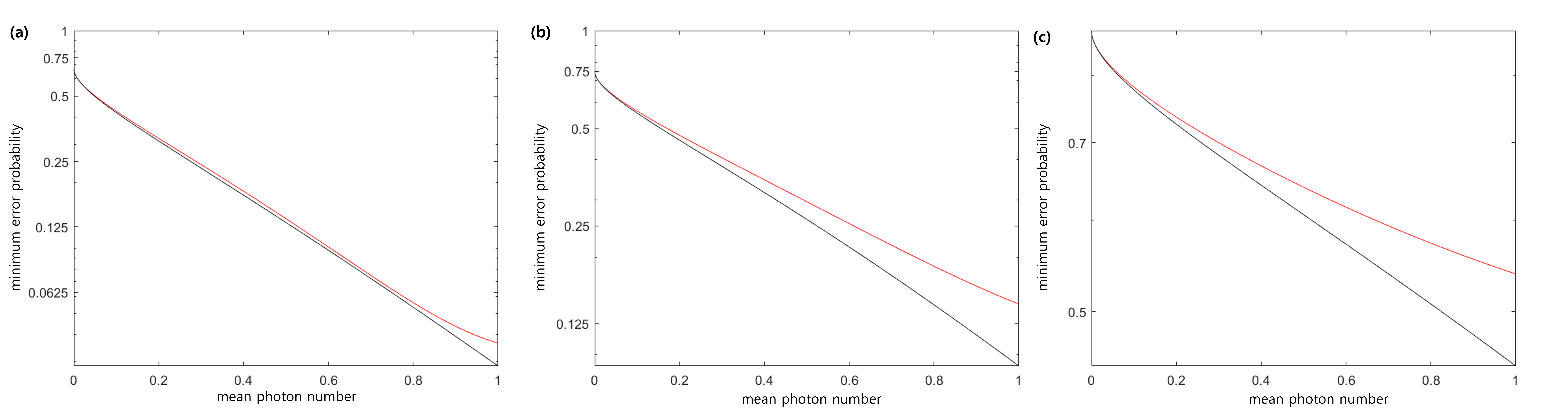}
\caption{The minimum error probabilities of $N$-SPSK signal and $N$-NSPSK signal composed of mSG-CS, where (a), (b) and (c) shows the case of $N=$3, 4 and 8, respectively. In these figures, red lines show the case of $N$-NSPSK signal and black lines show the case of $N$-SPSK. }
\end{figure*}
\begin{figure*}
\centering
\includegraphics[scale=0.29]{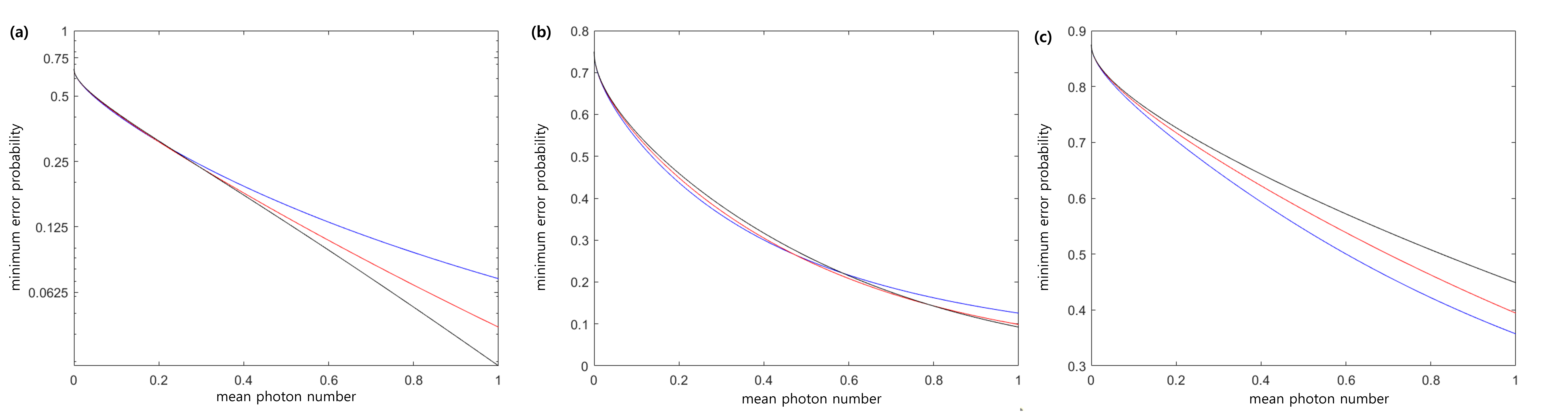}
\caption{The minimum error probabilities of $N$-SPSK signal and $N$-NSPSK signal composed of P-CS, where (a),  (b) and (c) shows the case of $N=$3, 4 and 8, respectively. In these figures, blue and red lines show the case of P-CS with $\varsigma=0.5$ and $\varsigma=1.5$, respectively. Black lines show the case of S-CS. }
\end{figure*}
The minimum error probabilities of $N$-SPSK signal and $N$-NSPSK signal composed of mSG-CS are illustrated in Fig. 3, where Fig. 3(a), (b) and (c) shows the case of $N=$3, 4 and 8, respectively. In these figures, red lines show the case of $N$-NSPSK signal and black lines show the case of $N$-SPSK. 

In Fig. 3, each minimum error probability of $N$-NSPSK signal is larger than that of $N$-SPSK signal for any $N=$3, 4 and 8 and any mean photon number. This result implies that \textit{3-PSK, 4-PSK, and 8-PSK quantum communication cannot be enhanced by mSG-CS.}  We also compare this result with the previous work about on-off keying signal \cite{e.m.f.curado}; it is known that the minimum error probability of on-off keying signal composed of mSG-CS has a singular point where the logarithm of the minimum error probability diverges to $-\infty$. This implies that the minimum error probability can achieve to zero. Unlike the result in the previous work \cite{e.m.f.curado}, the minimum error probabilities of 3, 4 and 8-NSPSK signal in Fig. 3 do not have such singular points.  

\subsection{Perelomov Coherent States (P-CS)}
The minimum error probabilities of $N$-SPSK signal and $N$-NSPSK signal composed of P-CS are illustrated in Fig. 4, where Fig. 4(a),  (b) and (c) shows the case of $N=$3, 4 and 8, respectively. In these figures, blue and red lines show the case of P-CS with $\varsigma=0.5$ and $\varsigma=1.5$, respectively. Black lines show the case of S-CS. 

In Fig. 4(a), each minimum error probabilty of 3-NSPSK signal composed of P-CS with $\varsigma=$0.5 and 1.5 is larger than that of 3-SPSK signal for arbitrary mean photon number. In other words, \textit{$3$-PSK quantum communication cannot be enhanced by non-standard coherent state using P-CS}. However, in Fig. 4(b), each minimum error probabilty of 4-NSPSK signal composed of P-CS is smaller than that of 4-SPSK signal when mean photon number is small ($\langle n\rangle<0.585$ and $\langle n\rangle<0.786$ in case of $\varsigma=0.5$ and $\varsigma=1.5$, respectively). This implies that \textit{4-PSK quantum communication can be enhanced by P-CS}. In Fig. 4(c), each minimum error probabilty of 8-NSPSK signal composed of P-CS is smaller than that of 8-SPSK signal for arbitrary mean photon number. This result is rather surprising since even super-Poissonianity in P-CS can enhance the 4-PSK and 8-PSK quantum communication unlike the binary case \cite{m.namkung}. We discuss the details in the next section.

\subsection{Mandel Parameter and $N$-NSPSK Quantum Communication}
It is known that sub-Poissonianity of non-classical light is one of the important statistical properties for improving Helstrom bound of binary quantum optical communication \cite{e.m.f.curado}. For this reason, we consider the following \textit{Mandel parameter},
\begin{equation}
Q_M^{(NS)}=\frac{(\Delta n)^2}{\langle n\rangle}-1,
\end{equation}
 where $\langle n\rangle$ is mean photon number and $\Delta n$ is standard deviation of the number of photons. It is known that if $Q_M^{(NS)}>0(<0)$, then the generalized coherent state is \textit{super-Poissonian}(\textit{sub-Poissonian}) \cite{l.mandel,r.short}. If $Q_M^{(NS)}=0$ (for example, S-CS), then the generalized coherent state is \textit{Poissonian}. Here, we consider the relation between the performance of the $N$-PSK quantum communication and the Mandel parameter.

\begin{enumerate}
\item In case of OS-CS, the Mandel parameter is analytically driven as \cite{e.m.f.curado}
\begin{eqnarray}
Q_M^{(OS)}=-\frac{\langle n\rangle}{\widetilde{n}},
\end{eqnarray}
which means that OS-CS is always sub-Poissonian. According to Fig. 1, we note that sub-Poissonianity of OS-CS does not always guarantee the enhancement of the $N$-PSK quantum communication.

\item In case of BG-CS, the Mandel parameter is analytically driven in terms of the Modified Bessel function of the first kind as \cite{e.m.f.curado}
\begin{eqnarray}
Q_M^{(BG)}=\alpha\left[\frac{I_{2\varsigma+1}(2\alpha)}{I_{2\varsigma}(2\alpha)}-\frac{I_{2\varsigma}(2\alpha)}{I_{2\varsigma-1}(2\alpha)}\right].
\end{eqnarray}
Since the inequality $\left\{I_{\nu+1}(x)\right\}^2\ge I_\nu(x)I_{\nu+2}(x)$ holds for every $x\ge 0$, $Q_M^{(BG)}$ is negative semidefinite. Therefore, BG-CS is always sub-Poissonian or Poissonian. Moreover, $Q_M^{(BG)}$ is known to be strictly negative for non-zero mean photon number \cite{e.m.f.curado}. Nevertheless, Fig. 2 shows that sub-Poissonianity of BG-CS does not always guarantee the enhancement of the $N$-PSK quantum communication.

\item Since the analytic form of the Mandel parameter of the mSG-CS Mandel parameter is too complex \cite{e.m.f.curado}, we do not introduce the analytic form here. According to the result of \cite{e.m.f.curado}, the Mandel parameter of mSG-CS is negative when the mean photon number is not too large. Nevertheless, Fig. 3 shows that the sub-Poissonianity of the mSG-CS cannot provide any advantage on the $N$-PSK quantum communication.

\item In case of P-CS, the Mandel parameter is analytically driven as \cite{e.m.f.curado}
\begin{eqnarray}
Q_M^{(P)}=\frac{\langle n\rangle}{2\varsigma},
\end{eqnarray}
which means that P-CS is super-Poissonian. However, Fig. 4 shows that P-CS can enhance the $N$-PSK quantum communication for $N=$3, 4 or 8. It is surprising since the super-Poissonianity of NS-CS can even enhance $N$-PSK quantum communication unlike the binary case.
\end{enumerate}

\section{Conclusion}
In the present article, we have considered the quantum communication with $N$-ary phase shift keying ($N$-PSK) signal for arbitrary an arbitrary positive integer $N>1$. By using NS-CS, we have analytically provided the Helstrom bound of the quantum communication with $N$-PSK. Unlike the binary case \cite{e.m.f.curado,m.namkung}, we have shown that even super-Poissonianity of NS-CS can improve the Helstrom bound of $N$-PSK quantum communication: The Helstrom bound can be improved by considering the sub-Poissonian NS-CS for $N=3$, meanwhile the super-Poissonian NS-CS can improve the Helstrom bound for $N=4$ and $N=8$.

Using $N$-PSK signal with $N>2$, we can achieve a better transmission rate per a signal pulse than that of binary-BPSK even if the receiver's measurement is slow \cite{i.a.burenkov}. On the other hands, the maximal success probability of discriminating a message encoded in $N$-PSK signal generally decreases as $N$ is getting large. Thus our results about the possible enhancement of the maximal success probability in $N$-PSK quantum communication by NS-CS is important and even necessary to design efficient quantum communication schemes.

In the present article, we have only considered PSK signal with equal prior probabilities, which is composed of symmetric pure states. However, it is interesting and even important to consider a non-equiprobable or asymmetric ensemble of NS-CS for several reasons: First, it is practically difficult to implement the PSK signal having perfect symmetry or equal prior probabilities. Moreover, in discriminating three non-equiprobable and asymmetric pure states, there is possibility that sub-Poissonianity of non-classical light can enhance the Helstrom bound. We note that it is also interesting to consider unambiguous discrimination \cite{i.d.ivanovic,d.dieks,a.peres,g.jaeger,s.pang,j.a.bergou2} of NS-CS since this strategy can provide us with better confidence than the minimum error discrimination.

\section{Acknowledgement}
This work was supported by Quantum Computing Technology Development Program (NRF2020M3E4A1080088)
through the National Research Foundation of Korea (NRF) grant funded by the Korea government (Ministry of
Science and ICT).


\begin{thebibliography}{99}
\bibitem{g.cariolaro} G. Cariolaro, \textit{Quantum Communications} (Springer, 2015).
\bibitem{c.w.helstrom} C. W. Helstrom, \textit{Quantum Detection and Estimation Theory} (Academic Press, 1976).
\bibitem{k.tsujino} K. Tsujino, D. Fukuda, G. Fujii, S. Inoue, M. Fujiwara, N. Takeoka, and M. Sasaki, ``Sub-shot-noise-limit discrimination of on-off keyed coherent states via a quantum receiver with a superconducting transition edge sensor'', Opt. Express \textbf{18}, 8107 (2010).
\bibitem{j.g.proakis} J. G. Proakis and M. Salehi, \textit{Digital Communications}, 5th ed. (McGraw-Hill, 2008). 
\bibitem{i.a.burenkov} I. A. Burenkov, M. V. Jabir, and S. V. Polyakov, ``Practical quantum-enhanced receivers for classical communication'', AVS Quantum Sci. \textbf{3}, 025301 (2021).
\bibitem{s.m.barnett} S. M. Barnett and S. Croke, ``Quantum state discrimination'', Adv. Opt. Photon. \textbf{1}, 238 (2009).
\bibitem{j.a.bergou} J. A. Bergou, ``Quantum state discrimination and selected applications'', J. Phys: Conf. Ser. \textbf{84}, 012001 (2007).
\bibitem{j.bae} J. Bae and L. C. Kwek, ``Quantum state discrimination and its applications'', J. Phys. A: Math. Theor. \textbf{48}, 083001 (2015).
\bibitem{d.ha} D. Ha and Y. Kwon, ``Complete analysis for three-qubit mixed-state discrimination'', Phys. Rev. A \textbf{87}, 062302 (2013).
\bibitem{s.j.dolinar} S. J. Dolinar, ``An optimum receiver for the binary coherent state quantum channel'', Q. Prog. Rep. \textbf{108}, 219 (1973).
\bibitem{e.m.f.curado} E. M. F. Curado, S. Faci, J.-P. Gazeau, and D. Noguera, ``Lowering Helstrom Bound with non-standard coherent states'', J. Opt. Soc. Am. B \textbf{38}, 3556 (2021).
\bibitem{m.namkung} M. Namkung and J. S. Kim, ``Indirect Measurement for Optimal Quantum Communication Enhanced by Binary Non-standard Coherent States'', arXiv:2112.02312.
\bibitem{s.dey} S. Dey, ``An introductory review on resource theories of generalized nonclassical light'', J. Phys: Conf. Ser. \textbf{2038}, 012008 (2021).
\bibitem{s.m.barnett2} S. M. Barnett and S. Croke, ``On the conditions for discrimination between quantum states with minimum error'', J. Phys. A: Math. Theor. \textbf{42}, 062001 (2009).
\bibitem{a.chefles} A. Chefles and S. M. Barnett, ``Optimum unambiguous discrimination between linearly independent symmetric states'', Phys. Lett. A \textbf{250}, 223 (1998).
\bibitem{r.j.glauber} R. J. Glauber, ``Coherent and Incoherent States of the Radiation Field'', Phys. Rev. \textbf{131}, 2766 (1963).
\bibitem{a.m.perelomov} A. M. Perelomov, \textit{Generalized Coherent States and Their Applications} (Springer, 1986).
\bibitem{a.o.barut} A. O. Barut and L. Girardello, ``New `coherent' states associated with non-compact groups'', Commun. Math. Phys. \textbf{21}, 2349 (2001).
\bibitem{j.-p.gazeau} J.-P. Gazeau, V. Hussin, J. Moran, and K. Zelaya, ``Generalized Susskind-Glogower coherent states'' J. Math. Phys. \textbf{62}, 072104 (2020).
\bibitem{l.mandel} L. Mandel, ``Fluctuations of Photon Beams: The Distribution of the Photo-Electrons'', Proc. Phys. Soc. \textbf{74}, 233 (1959).
\bibitem{r.short} R. Short and L. Mandel, ``Observation of Sub-Poissonian Photon Statistics'', Phys. Rev. Lett. \textbf{51}, 384 (1983).
\bibitem{i.d.ivanovic} I. D. Ivanovic, ``How to differentiate between non-orthogonal states'', Phys. Lett. A \textbf{123}, 257 (1987).
\bibitem{d.dieks} D. Dieks, ``Overlap and distinguishability of quantum states'', Phys. Lett. A \textbf{126}, 303 (1988).
\bibitem{a.peres} A. Peres, ``How to differentiate between non-orthogonal states'', Phys. Lett. A \textbf{128}, 19 (1988).
\bibitem{g.jaeger} G. Jaeger, ``Optimal distinction between two non-orthogonal quantum states'', Phys. Lett. A \textbf{197}. 83 (1995).
\bibitem{s.pang} S. Pang and S. Wu, ``Optimum unambiguous discrimination of linearly independent pure states'', Phys. Rev. A \textbf{80}, 052320 (2009).
\bibitem{j.a.bergou2} J. A. Bergou, U. Futschik, and E. Feldman, ``Optimal Unambiguous Discrimination of Pure Quantum States'', Phys. Rev. Lett. \textbf{108}, 250502 (2012).
\end{thebibliography}
\end{document}